\documentclass[runningheads,a4paper]{llncs}

\usepackage{amsmath}
\usepackage{amssymb}
\usepackage{array}

\usepackage{bm}

\DeclareMathOperator*{\minimum}{\text{Minimum}}

\usepackage[ruled, vlined]{algorithm2e}

\newtheorem{observation}{Observation}

\renewcommand{\c}[1]{\bm{#1}} 

\newcommand{\masterlp}{$(1)-(4)$}
\newcommand{\massys}{$(2)-(4)$}
\newcommand{\lp}{$(5)-(7)$}
\newcommand{\xlp}{$(5)-(8)$}
\newcommand{\lsip}{$(9)-(13)$}
\newcommand{\dwsub}{$(14)$}

\newcommand{\lscr}{$(10)-(13)$}
\newcommand{\bnlp}{$(15)-(18)$}
\newcommand{\mulp}{$(20)-(22)$}

\newcommand{\X}{{\footnotesize  \textbf{X}}}
\newcommand{\ax}{$\mathcal{A}$}

\newcommand{\bi}{{\scriptsize  BASIS INVERSE}}
\newcommand{\rhs}{{\scriptsize  RHS}}

\newcommand{\subprob}{{\scriptsize  SUBPROBLEM}}
\newcommand{\masprob}{{\scriptsize  MASTER PROBLEM}}

\newcommand{\dw}{DW}

\newcolumntype{P}[1]{>{\raggedleft\arraybackslash}p{#1}}

\begin{document}

\title{Mechanism Design via Dantzig-Wolfe Decomposition}

\author{Salman Fadaei \footnote{This work was done while the author was a graduate student in the Department of Informatics, TU M\"unchen, Munich, Germany.}}

\authorrunning{Fadaei} 
%
\tocauthor{Salman Fadaei}

\institute{ 
\email{salman.fadaei@gmail.com}}
\maketitle

\begin{abstract}

In random allocation rules, typically first an optimal fractional point is calculated via solving a linear program.
The calculated point represents a fractional assignment of objects or more generally packages of objects to agents.
In order to implement an expected assignment, the mechanism designer must decompose the fractional point into integer solutions, each satisfying underlying constraints.
The resulting convex combination can then be viewed as a probability distribution over feasible assignments out of which a random assignment can be sampled.
This approach has been successfully employed in combinatorial optimization as well as mechanism design with or without money.

In this paper, we show that both finding the optimal fractional point as well as its decomposition into integer solutions can be done at once.
We propose an appropriate linear program which provides the desired solution. 
We show that the linear program can be solved via Dantzig-Wolfe decomposition.
Dantzig-Wolfe decomposition is a direct implementation of the revised simplex method which is well known to be highly efficient in practice.
We also show how to use the Benders decomposition as an alternative method to solve the problem.
The proposed method can also find a decomposition into integer solutions when the fractional point is readily present perhaps as an outcome of other algorithms rather than linear programming.
The resulting convex decomposition in this case is tight in terms of the number of integer points according to the Carath{\'e}odory's theorem.


\end{abstract}

\section{Introduction}
The technique of finding a fractional solution, and decomposing it into polynomially-many integer points has been successfully employed in many problems.
For a usage of the technique in combinatorial optimization, for instance, see Carr and Vempala  \cite{carr2000randomized}.
In mechanism design with bidders who have quasi-linear valuations, the framework presented by Lavi and Swamy for designing truthful and approximate mechanisms strongly relies on this technique \cite{lavi2011truthful,lavi2005truthful}.
Perhaps the best connection between linear programming and algorithmic mechanism design has been established by this framework.
Finally, for applications in mechanism design without money see e.g. Budish et al.  \cite{budish2013designing}, and Nguyen et al. \cite{nguyen2015assignment}.

Typically in such applications, first a fractional optimal point is calculated, and in a second step, the point is represented as a convex combination of integer points usually using the ellipsoid method.
A subroutine or an approximation algorithm which returns an integer point with respect to any cost vector is employed to construct the separation oracle for the ellipsoid method.
A separation oracle, in the ellipsoid method, states if a given point is feasible, or in case it is not feasible, the oracle returns a violated constraint.
A natural question that arises here is to ask if the two steps, optimization as well as convex decomposition, can be done at once without employing the ellipsoid method?

\subsection{Results and Techniques}
We propose an appropriate linear program for finding an optimal fractional point and its decomposition into integer points. 
We show how to use the Dantzig-Wolfe decomposition which is based on the revised simplex to solve the linear program. 
More specifically, we show that finding a convex combination of integer points whose value is maximum is indeed equivalent to solving a linear program using the Dantzig-Wolfe decomposition. 
The proposed method will improve the connection between linear programming and algorithmic mechanism design.

Dantzig-Wolfe (DW) decomposition comprises a \emph{master problem} and a \emph{subproblem}.
DW decomposition proceeds in iterations by solving the two problems in each iteration until the subproblem is not able to  find any point which can contribute to the objective value of the master problem \cite{bazaraa2011linear}.
Since we are interested in integer points, we run a subroutine or an approximation algorithm which returns integer solutions as the subproblem.
DW decomposition has been previously used for optimizing over a discrete set using branch and cut to obtain integer solutions \cite{desrosiers2005primer}. 
However, here we assume the existence of a subroutine which returns an approximate integer solution of good quality and prove that the algorithm ends in optimality.
To the best of our knowledge, this usage of DW decomposition in mechanism design has not been introduced before.


Dantzig-Wolfe decomposition is a variant of the revised simplex algorithm.
A computational evaluation of the Dantzig-Wolfe decomposition has been done in \cite{tebboth2001computational}.
The study shows DW decomposition has a high performance, especially when a reasonable block structure can be found. 

The convex combination calculated by our method is tight in terms of the number of integer solutions according to the Carath{\'e}odory's theorem provided that the number of constraints representing the underlying polytope is less than the dimension of the polytope.
We explain this fact further in the following.

\begin{theorem}[Carath{\'e}odory]
Given a polytope in $\mathbb{R}^n$, any point in the polytope is a convex combination of at most $n+1$ vertices of the polytope.
\end{theorem}
By standard polyhedra theory, the number of nonzero variables in an extreme point is upper bounded by the number of constraints in the underlying linear program (see e.g. \cite{chvatal1983linear}).
The proposed algorithm produces a convex combination of at most $m+1$ integer points, where $m$ is the number of constraints, and thus the solution is tight in this sense. 
It is very common that the number of constraints $m$ is less than the number of variables $n$.
For example, in the relaxations of combinatorial auctions, this is usually the case because the bidders may obtain any package of items (for which one decision variable is needed) and there are exponentially many packages of items.
Thus,  given that $m < n$, the number of integer solutions will be at most $n+1$ which is tight according to the Carath{\'e}odory's theorem.

Sometimes, a fractional point - not necessarily an optimum - is calculated via other methods rather than linear programming.
For example, a greedy algorithm might be used to find a fractional point.
This is because \emph{truthfulness}\footnote{Truthfulness is a desired property in algorithmic mechanism design, and assures that no bidder would benefit from reporting false valuations.} can be guaranteed via the greedy algorithm, but directly solving the linear program cannot assure truthfulness (see e.g. \cite{dughmi2010truthful}).
Our method can find a decomposition into integer solutions for such readily-present fractional points.

We also show how to apply the Benders decomposition to the problem.
Benders decomposition is known to be the dual of the Dantzig-Wolfe decomposition technique \cite{bazaraa2011linear}.
We observe that sometimes working with the Benders decomposition has advantages over the DW decomposition.
We discuss these advantages further in a separate section of the paper.

\subsection{Related Literature}
Prior to this work, there have been other attempts to replace the ellipsoid in finding convex decompositions.
An alternative method is given by Kraft et al. \cite{kraft2014fast}.
The main component of that work is an algorithm which is based on a simple geometric idea that
computes a convex combination within an arbitrarily small distance $\epsilon > 0$ to the fractional point.
Our proposed method has advantages over the result in \cite{kraft2014fast}.
First, the size of the convex decomposition (number of integer solutions) is strictly smaller than the size of the convex decomposition produced by the method in \cite{kraft2014fast}.
The size of the convex decomposition in \cite{kraft2014fast} might be as large as $O(s^3\epsilon^{-2})$, where $s$ is the number of nonzero components of the fractional point, and $\epsilon>0$. 
Our solution will have a size of at most $s+1$.

Second, our decomposition is exact and does not suffer from an $\epsilon>0$ compromise in the solution.
However, we provide no theoretical upper bound on the number of iterations, and the proposed method relies on the performance of Dantzig-Wolfe decomposition in practice.

Elbassioni et al. present an alternative method for finding a convex decomposition of a given fractional point  \cite{elbassioni2015towards}.
Their method relies on the multiplicative weights update method which is a general technique for solving packing and covering problems \cite{arora2012multiplicative,khandekar2004lagrangian}.
While the algorithm presented in \cite{elbassioni2015towards} has a theoretical upper bound on the number of iterations, the algorithm is inferior to the presented method here in two aspects.
First, their convex decomposition might have a size (the number of integer solutions) of $s(\lceil {\epsilon^{-2} \ln s}\rceil+1)$, $s$ being the number of nonzero components of the fractional point, and $\epsilon>0$.
As mentioned earlier, our solution will have a size of at most $s+1$.
Second, their convex decomposition can be as precise as $\frac{1}{1+4\epsilon}$  times the fractional solution (for some $\epsilon>0$) at the expense of increasing runtime, while our convex decomposition is exact.

\subsection{Structure}
In Section \ref{sec-setting}, we formally introduce the setting of the problem.
In Section \ref{sec:dwd}, we provide a short summary of DW decomposition technique.
In Section \ref{sec-int-dw}, we establish our main result and show how the DW principle can be applied to our setting.
Section \ref{sec:benders} is devoted to the Benders decomposition applied to our setting.
Section \ref{sec:dw-apps} discusses two applications of the adapted DW principle.
Finally, in Section \ref{sec:example-dw}, we provide a numerical example for the adapted DW technique.

\section{Setting} \label{sec-setting}
Consider a finite set of integer points in $\mathbb{Z}_+^{n}$.
Let $Q$ denote the \textit{convex hull} of all these points. That is 
$Q$ defines a polytope with integral extreme points.
Let $P = \big\{\c{x} \in \mathbb{R}^n \ |\ \c{Ax} \leq \c{b}\ \ \& \ \c{x} \geq \c{0}\big\}$ denote a polytope,
where $\c{A}$ is an $m$ by $n$ matrix, and $\c{b}$ an $m$-dimensional column vector.
A subroutine $\mathcal{A}$ for any cost function $\c{c} \in \mathbb{R}^n$ returns an integer point $\X \in Q$ such that $\c{c}\X \geq \c{cx}^*$, where $\c{x}^*=\arg \max \big\{\c{cx} \ | \ \c{x} \in P\big\}$.
Equivalently, we say subroutine $\mathcal{A}$ will return for any cost vector $\c{c}$ an integer point $\X\in Q$ such that $\c{c}\X \geq \c{c} \c{x}$ for any $\c{x}$ in $P$.
Let $\mathcal{I}$ denote the index set for integer points in $Q$. The set of integer points in $Q$ is therefore $\big\{\X_j\big\}_{j \in \mathcal{I}}$. 

Usually, subroutine $\mathcal{A}$ only accepts non-negative cost vectors.
Examples are approximation algorithms for NP-hard optimization problems.
For instance, the approximation algorithm provided for the knapsack problem works with non-negative profits of items.
However, in our setting, we expect $\mathcal{A}$ to work with any arbitrary cost vector.
In such cases, an assumption that $Q$ is a packing polytope is required: if $x \in Q \ \text{and}\ y \leq x$ then $y \in Q$.
See Lavi and Swamy for more information  \cite{lavi2011truthful}.

In this paper, we address the following problem. 
Given a cost vector $\c{c} \geq \c{0}$, find values $\big\{\lambda^*_j \geq 0\big\}_{j \in \mathcal{I}}$ such that $i$) $\sum_{j \in \mathcal{I}}\lambda^*_j=1$, and $ii$) $|\{\lambda^*_j \ |\ j \in \mathcal{I},\ \lambda^*_j >0 \}|$ is polynomial in $m$ and $n$, and $iii$) $\sum_{j\in \mathcal{I}} \lambda^*_j \X_j = \c{x}^*$, where $\c{x}^*=\arg \max \big\{\c{cx} \ | \ \c{x} \in P\big\}$.

Using the ellipsoid method, it can be shown that every point in $P$ can be written as a convex combination of the extreme points in $Q$ \cite{carr2000randomized,lavi2011truthful}.
In this work, aside from answering the question above, we give an alternative proof for this fact. 

\section{Summary of Dantzig-Wolfe Decomposition} \label{sec:dwd}

Dantzig-Wolfe decomposition belongs to \textit{column generation} techniques.
We shall here briefly go over the Dantzig-Wolfe Decomposition. 
For a detailed explanation of the method we refer the reader to \cite{bazaraa2011linear}.
Consider the following linear program.
  \begin{alignat*}{2}
    \text{Minimize} \quad &  \c{cx}  \\
    \text{subject to} \quad & \c{Ax} = \c{b} \\
	 & \c{x} \in X &  
  \end{alignat*}

Where $X$ is a bounded polyhedral of special structure, $\c{A}$ is a $m\times n$ matrix, $\c{c}$ is a $n$-dimensional vector and $\c{b}$ is a $m$-dimensional vector.

Since $X$ is a bounded polyhedra, then any point $\c{x} \in X$ can be represented as a convex combination of a finite number of extreme points of $X$.
Let us denote these points by $\c{x}_1, \c{x}_2, \ldots, \c{x}_l$, and substitute $\c{x}$ with its convex combination of extreme points, then the aforementioned LP can be transformed into the following program in which the variables are $\lambda_1,\lambda_2,\ldots, \lambda_l$.

  \begin{align}
    \text{Minimize} \quad &  \sum_{j=1}^l (\c{cx}_j)\lambda_j  \\
	 \text{subject to} \quad & \sum_{j=1}^l (\c{Ax}_j)\lambda_j = \c{b} \label{dw-cons1} \\
	 &\sum_{j=1}^l \lambda_j=1,   \label{dw-cons2} \\
	 &\lambda_j \geq 0 \quad j=1,2,\ldots,l  
  \end{align}
  
The linear program \masterlp\ is called the \textit{master problem} and the program which finds an appropriate $\c{x} \in X$ in each iteration is called the \textit{subproblem}. 
Since the number of extreme points of set  $X$ is exponentially many, we follow the idea of \textit{column generation} to find appropriate extreme point in each iteration.
The information is passed back and forth between the master problem and the subproblem as follows.
In each iteration a different cost coefficient is passed down by the master problem to the subproblem and the subproblem finds an extreme point $\c{x}_k$, and sends it to the master problem.

Dantzig-Wolfe decomposition is an implementation of the \textit{revised simplex} method.
Let vector $\c{w}$ and $\alpha$ denote the dual variables corresponding to equations $(\ref{dw-cons1})$ and $(\ref{dw-cons2})$, respectively.
We first need an initial solution to generate the simplex tableau. 
Suppose we have a basic feasible solution $\bm{\lambda}=(\bm{\lambda}_B,\bm{\lambda}_N)$ to system \massys, where $\bm{\lambda}_B$ and $\bm{\lambda}_N$ denote the basic and nonbasic variables, respectively.
The initial $(m+1)\times (m+1)$ basis inverse $\textbf{B}^{-1}$ hence will be known. 
The cost for each basic variable $\lambda_j$ is in fact $\hat{c}_j=\c{cx}_j$.
Therefore, we get $(\c{w},\alpha)=\c{\hat{c}_B} \c{B}^{-1}$, where $\c{\hat{c}_B}$ is the cost vector of the basic variables. Denoting $\c{\bar{b}}=\c{B}^{-1} \begin{pmatrix} \c{b}\\ 1 \end{pmatrix}$, we see the revised simplex tableau in Table \ref{tbl-simplex}.
  
\begin{table}
\centering
\setlength{\extrarowheight}{5.0pt}
\begin{tabular}{  m{1cm} | >{\centering\arraybackslash} p{1cm} |>{\centering\arraybackslash} p{1cm} |}
\multicolumn{1}{r}{}
 &  \multicolumn{1}{c}{BASIS}
 & \multicolumn{1}{c}{RHS} \\
\cline{2-3}
 & $(\c{w},\alpha)$ & $\c{\hat{c}_B} \c{\bar{b}}$ \\
\cline{2-3}
& $\c{B}^{-1}$ & $\c{\bar{b}}$ \\
\cline{2-3}
\end{tabular}
\caption{Simplex tableau. RHS stands for right-hand side.}
\label{tbl-simplex}
\end{table}  

The revised simplex proceeds by improving the current solution via finding an entering and a leaving variable. 
In other words, the set of basic and nonbasic variables exchange one element.
When such an exchange is not possible then the current solution is optimal.
The entering variable is in fact a variable $\lambda_k$ associated with extreme point $\c{x}_k$ for which $z_k-\hat{c}_k>0$, where 
$z_k=(\c{w},\alpha)  \begin{pmatrix} \c{Ax}_k\\ 1 \end{pmatrix}$ and $\hat{c}_k=\c{cx}_k$.

We observe that $z_k-\hat{c}_k=(\c{wA}-\c{c})\c{x}_k+\alpha$ denotes the value of point $\c{x}_k$ with respect to current costs $\c{wA}-\c{c}$ and dual variable $\alpha$.
In order to find such a point, we solve the following subproblem which gives us the required index or tells that the current solution is optimal when the maximum value is zero.

  \begin{alignat*}{2}
    \text{Maximize} \quad &  (\c{wA}-\c{c})\c{x} + \alpha  \\
    \text{subject to} \quad & \c{x} \in X 
  \end{alignat*}
  
Notice that the objective function contains a constant and therefore it can be replaced by $(\c{wA}-\c{c})\c{x}$. 
Assuming that $\c{x}_k$ is the optimal solution to the program above, the revised simplex method goes on as follows. If $z_k-\hat{c}_k=0$ then the algorithm stops and the last solution to the master problem is an optimal of the overall problem. 

If $z_k-\hat{c}_k>0$ the master problem proceeds as follows. 
Let $\c{y}_k=\c{B}^{-1}\begin{pmatrix} \c{Ax}_k\\ 1 \end{pmatrix}$, the entering column then will be $\begin{pmatrix} z_k-\hat{c}_k\\ \c{y}_k \end{pmatrix}$.
In order to find the leaving column, let index $r$ be determined as follows:

\begin{displaymath}
\displaystyle \frac{\bar{b}_r}{y_{rk}}=\displaystyle \minimum_{1\leq i\leq m+1}\Bigg\{\frac{\bar{b}_i}{y_{ik}}:y_{ik}>0\Bigg\}.
\end{displaymath}
We pivot at $y_{rk}$ which will update the dual variables, the basis inverse, and the right-hand side.
More specifically, \emph{pivoting} on $y_{rk}$ can be stated as follows.

\begin{enumerate} \label{pivoting}
\item Divide row $r$ by $y_{rk}$.

\item For $i=1, \ldots m$ and $i \neq r$, update the $i$th row by adding to it $-y_{ik}$ times the new $r$th row.

\item Update row zero by adding to it $z_k-\hat{c}_k$ times the new $r$th row.

\end{enumerate}
After pivoting, the column $\lambda_k$ is deleted and the algorithm repeats.

An important observation about the DW principle is as follows.
The DW principle expects the Subproblem to return any point $\c{x} \in X$ where $(\c{wA}-\c{c})\c{x}+\alpha > 0$ and it does not impose any other specific requirement on the selected point.
We shall use this observation in our adaptation of the method.

Another observation is that, in each iteration, the master program finds the best solution using known extreme points. 
This is done in an organized manner as described above.

\section{Applying Dantzig-Wolfe Decomposition} \label{sec-int-dw}

A wide range of combinatorial optimization problems can be formulated using the integer program  
$\max \big\{\c{cx}\ |\ \c{x}\in P\,  \text{and integer}\big\}$.
Recall, $P = \big\{\c{x} \in \mathbb{R}^n \ |\ \c{Ax} \leq \c{b}\ \ \& \ \c{x} \geq \c{0}\big\}$.
For example, in combinatorial auctions, $\c{c}$ denotes the accumulated valuations of the players, and the integer program therefore expresses the welfare maximization objective subject to feasibility constraints encoded as $P$.

Usually, using simplex method or other standard linear programming techniques, first a relaxed linear program of the integer program above is solved:

  \begin{align}
    \text{Maximize} \quad &  \c{cx}  \\
	 \text{subject to} \quad &\c{Ax} \leq \c{b} \ &   \label{lp-cns}\\
	 & \c{x} \geq \c{0} &    \label{lp-cns-pos}
  \end{align}
  
Notice, constraints (\ref{lp-cns}) and (\ref{lp-cns-pos}) together are equivalent to $\c{x} \in P$.  
Next, the solution is rounded to an integer solution at the expense of a value loss or slightly violating the constraints.
Given subroutine \ax, we wish to find a solution to the linear program above as well as a convex decomposition of it into integer points. 
We wish to achieve both goals at once.
Recall, subroutine \ax\ returns for any cost vector $\c{c}$ an integer point $\X\in Q$ such that $\c{c}\X \geq \c{c} \c{x}$ for any $\c{x}$ in $P$.

Generally speaking, Dantzig-Wolfe principle uses the fact that if we relax some constraints and obtain a simpler polyhedra then the solution to the original problem can be written as a convex combination of extreme points of the simpler polyhedra. The simplicity refers to the fact that the extreme points of the new polyhedra can be found more easily than those of the original problem.
While Dantzig-Wolfe principle is useful when the underlying constraints are decomposable into simpler regions, we use it in a slightly different manner by looking at $Q$ as the polyhedra over which we can efficiently optimize.
Recall, $Q$ denotes the convex hull of a finite set of integer points.
To apply the idea of Dantzig-Wolfe decomposition to the problem \lp, we add a new constraint to the program.
  \begin{align}
	 & \c{x} \in Q &    \label{lp-cns-q}
  \end{align}
  
We will show that an optimal solution to problem \xlp\ is also an optimal solution to \lp, and thus adding the new constraint is harmless.
We represent $\c{x} \in Q$ as a convex combination of extreme points of $Q$.
Recall, $\mathcal{I}$ denote the index set for integer points in $Q$, and $\big\{\X_j\big\}_{j \in \mathcal{I}}$ is the set of integer points in $Q$. 
Substitute $\c{x}$ with its convex combination of extreme points of $Q$, then program \xlp\ can be transformed into the following program where the variables are $\big\{\lambda_j\big\}_{j \in \mathcal{I}}$.

  \begin{align}
    \text{Maximize} \quad &  \sum_{j \in \mathcal{I}} (\c{c}\X_j)\lambda_j  \\
	 \text{subject to} \quad &\sum_{j \in \mathcal{I}} (\c{A}\X_j)\lambda_j + \c{s} = \c{b}\\	 
	 &\sum_{j \in \mathcal{I}} \lambda_j=1  \\
	& \lambda_j \geq 0 \quad \forall j\in \mathcal{I} \\
	 &\c{s} \geq \c{0}.  
  \end{align}

The linear program \lsip\ is the \textit{master problem}. 
Notice, we have added slack variables $\c{s} \in \mathbb{R}_+^m$ to convert the inequalities into equality as needed by the \dw\ principle (see Section \ref{sec:dwd}).
The \dw\ \textit{subproblem} is defined below. 

  \begin{align} 
    \max_{j \in \mathcal{I}} \ &  (\c{c}+\c{wA})\X_j  
  \end{align}

While our problem is a maximization problem, the procedure in Section \ref{sec:dwd} is presented for a minimization problem, thus we substitute the $-\c{c}$ with $\c{c}$ in the objective function of the subproblem. 
That means, we change the objective function of the master problem from $\max \sum_{j \in \mathcal{I}} (\c{c}\X_j)\lambda_j$ to $\min\sum_{j \in \mathcal{I}} (-\c{c}\X_j)\lambda_j$, and apply the theory provided in Section \ref{sec:dwd}.

We assume $\c{0} \in Q$ thus the initial basic solution is simply defined by letting $\lambda_0$, the variable corresponding to point $\X=\c{0}$, equal $1$ ($\lambda_0=1$), and letting $\c{s}=\c{b}$.\footnote{Note that $\c{0} \in Q$ is a consequence of the assumption that $Q$ is a packing polytope.} 
Assuming that program \dwsub\ can be efficiently solved for any cost vector, then we are exactly following the DW principle, and therefore, we can successfully solve the overall problem as DW principle does this.
However, program \dwsub\ is an integer program and solving it may not be computationally tractable.

To address this issue, we propose using subroutine $\mathcal{A}$ to approximate program \dwsub.
Let $(\bar{\c{w}},\bar{\alpha})$ be the last dual variables calculated by the master program.
In each iteration, we call subroutine $\mathcal{A}$ with current cost vector $\c{c}+\bar{\c{w}}\c{A}$ to find a point $\X_k \in Q$ to pass to the master problem.
This substitution seemingly comes at the expense of stopping at a local optimum, as explained in the following.

As long as the algorithm continues by using the points returned by subroutine $\mathcal{A}$, we are exactly running DW principle. 
Recall the important observation that in DW principle, the subproblem need not be completely optimized and any point $\X_k$ with $(\c{c}+\bar{\c{w}}\c{A})\X_k+\bar{\alpha} > 0$ suffices to proceed.
However, there might be an iteration in which there exists a point $\X' \in Q$ for which we have $(\c{c}+\bar{\c{w}}\c{A})\X'+\bar{\alpha} >0$, but for the integer point $\X_k$ returned by subroutine $\mathcal{A}$, we have $(\c{c}+\bar{\c{w}}\c{A})\X_k+\bar{\alpha} \leq 0$. 
Therefore, DW stops at a suboptimal point.
Nevertheless, below, we argue that this cannot happen. 
That means as long as DW has not reached the optimum to problem \lp, subroutine $\mathcal{A}$, given the current cost vector $(\c{c}+\bar{\c{w}}\c{A})$, returns a point $\X_k$ with $(\c{c}+\bar{\c{w}}\c{A})\X_k+\bar{\alpha} >0$.

Let $\c{x}^*$ denote the optimal solution to program \lp.
It is instructive to see what the master step would do if the subproblem passes the point $\c{x}^*$, rather than an integer point, to the master problem.
While we do not know such an optimal point, but we know that such a point exists and this suffices for our reasoning.
We argue that the master step will set $\c{\lambda}_{x^*}=1$ and $\c{\lambda}_j=0$ for all other $\c{\lambda}_j$'s which are currently in the base.
In other words, the master program returns the best possible convex combination which is in fact $\lambda_{x^*}=1$. 
This is discussed in the following observation.

\begin{observation} \label{obs-mas-x-star}
Let $\c{x}^*$ denote the optimal solution to program \lp.
If supposedly the subproblem in any iteration passes $\c{x}^*$ to the master problem, then the master step will set $\c{\lambda}_{x^*}=1$ and $\c{\lambda}_j=0$ for all other $\c{\lambda}_j$'s which are currently in the base.
\end{observation}
\begin{proof}
First, we observe that if 
the first subproblem (right after the initialization) passes $\c{x}^*$ to the master problem, then the master step will set $\c{\lambda}_{x^*}=1$ and $\c{\lambda}_0=0$.
Clearly, in the solution $\c{s}$ needs to be evaluated accordingly.
Second, by looking more closely at what simplex does in each iteration, we observe that in any further iteration, if the subproblem passes $\c{x}^*$ to the master problem, the master step will set $\c{\lambda}_{x^*}=1$ and $\c{\lambda}_j=0$ for all other $\c{\lambda}_j$'s. 

The simplex method, in each iteration, performs a set of row operations on the constraints when it pivots (see pivoting steps in Section \ref{sec:dwd}).
The constraints, in any iteration, are thus the initial constraints after a series of row operations.
This will certify that the aforementioned solution ($\c{\lambda}_{x^*}=1$) will be feasible in any further iteration.
If the subproblem passes $\c{x}^*$ to the master problem, our entering variable will be $\c{\lambda}_{x^*}$.
The simplex algorithm then increases the entering variable $\c{\lambda}_{x^*}$ as much as one basic variable gets zero.
However, as discussed, the solution $\c{\lambda}_{x^*}=1$ is feasible, and it is possible to increase $\c{\lambda}_{x^*}$ up to $1$ and set all other $\c{\lambda}_j$'s to zero.
The algorithm will behave as such to produce the highest increase in the objective value, the desired conclusion.
\end{proof}

\begin{theorem} \label{lem-dw-non-stop}
If the subproblem \dwsub\ calls subroutine \ax\ to return an integer point in each iteration, the \dw\ principle never stops until it gets to an optimal solution to problem \lp.
\end{theorem}
\begin{proof}

Let $\c{x}^*$ denote the optimal solution to program \lp.
Assume the algorithm stops at a suboptimal point: $\sum_{j \in \mathcal{I}} (\c{c}\X_j)\lambda_j < \c{c}\c{x}^*$. 
Let $(\bar{\c{w}},\bar{\alpha})$ be the last dual variables calculated by the master program.
Let $\X_k$ be the point returned by subroutine $\mathcal{A}$, given cost vector $(\c{c}+\bar{\c{w}}\c{A})$, in the last iteration.
We must have $(\c{c}+\bar{\c{w}}\c{A})\X_k+\bar{\alpha} \leq 0$ because \dw\ has stopped.

If supposedly the subproblem in the last iteration passes $\c{x}^*$ to the master problem, according to Observation \ref{obs-mas-x-star}, the master step will increase $\c{\lambda}_{x^*}$ up to $1$ and set all other $\c{\lambda}_j$'s to zero.
Since we assumed the algorithm has stopped at a suboptimal point,  by setting $\c{\lambda}_{x^*}=1$ the objective value will increase. 
If entering $\c{\lambda}_{x^*}$ improves the objective value, we must have  $(\c{c}+\bar{\c{w}}\c{A})\c{x}^*+\bar{\alpha} > 0$ from the theory of \dw\ principle provided in Section \ref{sec:dwd}: if  $(\c{c}+\bar{\c{w}}\c{A})\c{x}^*+\bar{\alpha} \leq 0$ then entering variable $\c{\lambda}_{x^*}$ cannot improve the objective value.

By the property of the subroutine, we have $(\c{c}+\bar{\c{w}}\c{A})\X_k \geq (\c{c}+\bar{\c{w}}\c{A}) \c{x^*}$.
Thus, in the last iteration, we must have $(\c{c}+\bar{\c{w}}\c{A})\X_k+\bar{\alpha} > 0$. 
This contradicts our assumption that $(\c{c}+\bar{\c{w}}\c{A})\X_k+\bar{\alpha} \leq 0$.
Consequently, as long as we have not reached the optimum, the subroutine returns a point  $\X_k$ with $(\c{c}+\bar{\c{w}}\c{A})\X_k+\bar{\alpha} > 0$.
This completes the proof.
\end{proof}

We draw the conclusion that substituting program \dwsub\ with subroutine $\mathcal{A}$ is harmless.
Therefore, we have shown that finding a convex decomposition of maximum value is indeed equivalent to solving a linear program via \dw\ principle.
Let us call the method \textit{integer \dw}.

\section{Benders Decomposition} \label{sec:benders}
It is known that the Dantzig-Wolfe Decomposition has an equivalent decomposition technique namely \emph{Benders decomposition} \cite{bazaraa2011linear}.
Benders decomposition is a \emph{row generation technique} in contrast with the Dantzig-Wolfe column generation procedure.
Sometimes, working with Benders decomposition has advantages over Dantzig-Wolfe decomposition.
We explain how to apply the Benders algorithm to our problem.
Later, we discuss the advantages of the method.

Recall, polytope $Q$ is a bounded polyhedra.
Hence, there exist matrix $\c{D} \in \mathbb{R}^{m' \times n}$ and $\c{d} \in \mathbb{R}^{m'}$ such that $Q = \big\{\c{x} \in \mathbb{R}^n \ |\ \c{Dx} \leq \c{d}\ \ \& \ \c{x} \geq \c{0}\big\}$.
We add constraint $\c{x} \in Q$ to program \lp\ and work with the new program.
We will see that adding this constraint has no influence on the region of feasible solutions to program \lp.
Following the standard procedure \cite{bazaraa2011linear}, we can write the Benders decomposition for this new program. 
The \emph{Benders master problem} will be as follows.
  \begin{align}
    \text{Maximize} \quad &  z  \\
	 \text{subject to} \quad &z \leq \c{wb}-(\c{c}+\c{wA})\X_j \quad \forall j\in \mathcal{I} &  \label{bms-int-cns} \\
	 & \c{w} \leq \c{0} \\
	 & z \quad \text{unrestricted.} &    
  \end{align}

The variables of the master problem are $z$ and $\c{w}$.
Variable vector $\c{w}$ is the vector of dual variables associated to the constraints (\ref{lp-cns}).
The Benders master problem has exponentially many constraints, thus it is inconvenient to solve directly.
Hence, we maintain only a few of the constraints (\ref{bms-int-cns}).
Assuming $\c{0} \in Q$, we start with only one constraint: $z \leq \c{wb}-(\c{c}+\c{wA})\c{0}=\c{wb}$.
Notice, we can use any $\X_j \in Q$ to start with.
We solve the master problem and let $(\bar{z},\bar{\c{w}})$ be the solution.
The value of $\bar{z}$ is an upper bound on the optimal value to the master problem.
If $(\bar{z},\bar{\c{w}})$ satisfies constraints (\ref{bms-int-cns}) for all $j \in \mathcal{I}$, then $(\bar{z},\bar{\c{w}})$ is optimal for the master problem.
We can check constraints (\ref{bms-int-cns}) by examining if $\bar{z} \leq \bar{\c{w}}\c{b}-\max_{j \in \mathcal{I}}(\c{c}+\bar{\c{w}}\c{A})\X_j$.

Thus, the \emph{Benders subproblem} will be as the following.
  \begin{align} \label{benders-sub}
    \max_{j \in \mathcal{I}} \ &  (\c{c}+\bar{\c{w}}\c{A})\X_j  
  \end{align}

Note that the Benders subproblem is also the subproblem solved by the Dantzig-Wolfe decomposition.
Furthermore, the Benders master problem is the dual to the Dantzig-Wolfe master problem \lsip.
The Benders subproblem, in each iteration, is solved by calling subroutine $\mathcal{A}$ with cost vector $\c{c}+\bar{\c{w}}\c{A}$.
If the subproblem returns $\X_k$ that violates the constraints (\ref{bms-int-cns}): $\bar{z} > \bar{\c{w}}\c{b}-(\c{c}+\bar{\c{w}}\c{A})\X_k$, we can generate the new constraint $z \leq \c{w}\c{b}-(\c{c}+\c{w}\c{A})\X_k$, and add it to the current master program, and reoptimize.
We repeat this process until the solution returned by  subroutine $\mathcal{A}$ does not violate the constraints.
We claim that at this iteration, the value of $\bar{z}$ is the optimal value to the master problem.

The Benders decomposition provides a more concise proof that the decomposition techniques in companion with the subroutine \ax\ work correctly.

\begin{theorem} \label{lem-dw-non-stop}
If the subproblem (\ref{benders-sub}) calls subroutine \ax\ to return an integer point in each iteration, the Benders algorithm never stops until it reaches an optimal solution to problem \lp.
\end{theorem}
\begin{proof}

Let $\c{x}^*$ denote an optimal solution to problem \lp.
Let $z^*$ denote an optimal value of the Benders master problem.
We have $z^* = (-\c{c})\c{x}^*$ by the construction of the Benders master problem, and the duality theorem.
Let $\bar{z}$ be the final solution to the master problem and $\X_k$ the solution returned by subroutine \ax\ when the Benders algorithm stops.
Since the algorithm stops, we must have $\bar{z} \leq \bar{\c{w}}\c{b}-(\c{c}+\bar{\c{w}}\c{A})\X_k$.
We always have $z^* \leq \bar{z}$ because $\bar{z}$ is an upper bound on the optimal solution to the master problem.
Assume $\bar{z}$ is not optimal: $\bar{z} > z^*$.
Remember, by the definition of subroutine \ax, we have $(\c{c}+\bar{\c{w}}\c{A})\X_k \geq (\c{c}+\bar{\c{w}}\c{A})\c{x}^*$. 
Therefore,

\begin{displaymath}
\begin{array}{llll}
& \c{A}\c{x}^* & \leq \c{b} & \text{since}\ \c{x}^*\ \text{is a solution to}\\

&  & & \text{problem \lp}\\

\Rightarrow{} & \bar{\c{w}}\c{A}\c{x}^* & \geq \bar{\c{w}}\c{b} & \text{since}\  \bar{\c{w}} \leq \c{0}\\

\Rightarrow & -\c{c}\c{x}^* & \geq \bar{\c{w}}\c{b} -\c{c}\c{x}^* -\bar{\c{w}}\c{A}\c{x}^* & \\

\Rightarrow & -\c{c}\c{x}^* & \geq \bar{\c{w}}\c{b} -(\c{c}+\bar{\c{w}}\c{A})\X_k & \text{since}\ (\c{c}+\bar{\c{w}}\c{A})\X_k \geq (\c{c}+\bar{\c{w}}\c{A})\c{x}^*   \\

\Rightarrow & \bar{z} & > \bar{\c{w}}\c{b} -(\c{c}+\bar{\c{w}}\c{A})\X_k & \text{since}\ \bar{z} > z^* = -\c{c}\c{x}^*

\end{array}
\end{displaymath}

But, this contradicts $\bar{z} \leq \bar{\c{w}}\c{b}-(\c{c}+\bar{\c{w}}\c{A})\X_k$. 
The contradiction arises from the assumption that $\bar{z}$ is not optimal. 
Thus, when the algorithm stops, we have the optimal solution, the desired conclusion. 
\end{proof}

After solving the Benders master problem, we can use the provided integer solutions and solve the restricted primal to obtain a convex decomposition.

Another advantage of the Benders decomposition arises from the fact that in each iteration, we optimize an LP in the master step.
Solving an LP is sometimes more convenient than the implementation of the pivoting steps done in each iteration in the master step of the \dw\ principle.

\subsection{Polynomial Runtime Using Ellipsoid}
It is instructive to note that Theorem \ref{lem-dw-non-stop} implies that if the integer solution returned by subroutine $\mathcal{A}$ does not violate constraints (\ref{bms-int-cns}), then the current master solution is optimal.
Exploiting this fact, we can use the ellipsoid method to solve problem \bnlp\ to certify a polynomial runtime which might be of theoretical interest.
To use the ellipsoid method, we need to implement a \emph{separation oracle}.
Recall that a separation oracle, given a solution, either confirms that it is a feasible solution, or returns the constraint violated by the solution.
Using subroutine $\mathcal{A}$ as the separation oracle, as long as we find a violated constraint, we cut the current ellipsoid and continue.
When the subroutine $\mathcal{A}$ cannot return a violating constraint, according to Theorem \ref{lem-dw-non-stop}, the algorithm has reached the optimum.

\section{Application of the Method in 
Mechanism Design} \label{sec:dw-apps}
\subsection{The Framework Proposed by Lavi and Swamy}
Let $X = \big\{\c{x} \in \mathbb{R}^n \ |\ \c{Ax} \leq \c{b}\ \&\ \c{x} \geq \c{0}\big\}$ denote the underlying polytope of a linear program, and $\c{x}^*$ denote an optimal solution to the program with respect to some cost vector.
The maximum ratio between the value of an integer program and its relaxation, with respect to all cost vectors, is called the integrality gap of the relaxation.
Assuming that, the integrality gap of $X$ is $\beta \geq 1$, and that a $\beta$ integrality-gap-verifier is given, Lavi and Swamy propose a method to decompose the scaled-down fractional solution $\frac{\c{x}^*}{\beta}$ into a convex combination of integer solutions \cite{lavi2011truthful}.
A $\beta$ integrality-gap-verifier is an algorithm that, given any cost vector, returns an integer solution whose value is at least $1/\beta$ times the optimal relaxed solution.

This decomposition technique was originally observed by Carr and Vempala \cite{carr2000randomized}, and later adapted by Lavi and Swamy to mechanism design problems provided that the underlying polytope of the relaxation of the problem has the packing property \cite{lavi2011truthful}.
The approach requires only a polynomial number of calls to the integrality-gap-verifier with respect to the number of positive components in $x^*$.
Yet, the approach strongly relies on the ellipsoid method, and hence it is more of theoretical importance than of practical use.

In order to view the LS framework in our setting, the integrality-gap-verifier is used as subroutine $\mathcal{A}$ and $X/\beta = \big\{\c{x}\ |\ \beta \c{x} \in X\big\}$ is treated as $P$ in our setting introduced in Section \ref{sec-setting}.
This way, the integer DW finds the maximum value in $X/\beta$ as well as its decomposition into integer points, both in one step.
This improves upon other implementations of the LS framework which require two steps to find the convex decomposition  \cite{lavi2011truthful,kraft2014fast,elbassioni2015towards}.

It is instructive to note that solving program \lsip, essentially defines a Maximal-In-Distributional-Range (MIDR) allocation rule.
An MIDR algorithm fixes a set of distributions over feasible solutions (the distributional range) independently of the valuations reported by the self-interested players, and outputs a random sample from the distribution that maximizes expected (reported) welfare \cite{Dobzinski09}.
Here, we optimize over a range which is independent of bidder's private information. 
The range is in fact the feasible region of the program: all probability distributions over integer solutions which satisfy constraints \lscr.
The range is obviously independent of bidders' valuations.

\subsection{Existing Fractional Point}
Sometimes a fractional point $\c{x}^* \in Q$ is present, and we wish to find a convex decomposition of $\c{x}^*$ into extreme points of $Q$.
This can happen when we use other methods to find a fractional point rather than linear programming.
Recall, we assume that $Q$ satisfy the packing property.

For this case, we can use the integer DW as follows.
Define $P=\big\{\c{x} \in \mathbb{R}^n\ |\ \c{x}\leq\c{x}^* \ \&\ \c{x} \geq 0 \big\}$ and let $\c{c}=\c{x}^*$.
Now, apply the integer DW. 
All arguments follow accordingly, assuming that a subroutine $\mathcal{A}$ with the following property is available.
Subroutine $\mathcal{A}$ will return for any cost vector $\c{c}$ an integer point $\X\in Q$ such that $\c{c}\X \geq \c{c} \c{x^*}$.
Because the number of constraints in $P$ is at most $n$, the resulting convex decomposition in this case is tight in terms of the number of integer points, according to the Carath{\'e}odory's theorem.


\section{Numerical Example for Integer DW} \label{sec:example-dw}
In this section, we apply the integer DW to an instance of multi-unit auctions to see how the method works. 
We relegate the details to the appendix.

\bibliographystyle{splncs03}
\bibliography{literature}

\appendix

\section{Numerical Example for Integer DW} \label{sec:example-dw-apx}
In this section, we focus on applying the integer DW to an instance of multi-unit auctions. 
In multi-unit auctions, there is a set of $m$ identical items and a set of players. 
Each player $i$ has a valuation for any number of items denoted by $v_i(j)$ for getting $j$ items where $1\leq j\leq m$.
The goal is to maximize social welfare by distributing items among bidders.

The LP relaxation for this class of problems is as follows. 
Let $x_{ij}$ denote if $j$ units is assigned to bidder $i$.

  \begin{align}
    \text{Maximize} \quad &  \sum_{i,j} v_i(j) x_{ij} \tag{MU-P}  \\
	 \text{subject to} \quad &\sum_j x_{ij} \leq 1 \quad \text{for each player } i   \\
	 & \sum_{i,j} j\cdot x_{ij} \leq m     \\
	 & 0\leq x_{ij}\leq 1  \quad \text{for each } i,j    
  \end{align}
  
Lavi and Swamy present a greedy algorithm which returns for any valuation $v$ an integer solution that is at least as good as half of the optimal fractional solution to $\text{MU-P}$ with respect to $v$  \cite{lavi2011truthful}.
Thus, we have a $2$ integrality-gap-verifier algorithm for $\text{MU-P}$.
This greedy algorithm will serve as the subroutine in the integer DW, and is called \ax.

We give a short example to demonstrate the proposed convex decomposition method. 
Suppose a simple multi-unit auction with 3~players and 4~identical items. 
The following valuation vectors $v_i(j)$ are given for each player $i$ and quantity $j$:

\begin{align*}
	\begin{array}{lrrrrrr}
	j                       & 1  & 2  & 3  & 4  &  \\
	v_1(j) = \left( \right. & 6  & 6  & 6  & 6  & \left. \right) \\
	v_2(j) = \left( \right. & 1  & 4  & 4  & 6  & \left. \right) \\
	v_3(j) = \left( \right. & 0  & 1  & 1  & 1  & \left. \right)
	\end{array}
\end{align*}

We can reproduce program \lsip\ for this instance as follows. 
Let $\mathcal{I}$ denote the index set of integer points which satisfy inequalities \mulp.\\
Let 
$\setcounter{MaxMatrixCols}{20}$
$\c{c}=\begin{bmatrix} 6 & 6 & 6 & 6 & 1 & 4 & 4 & 6 & 0 & 1 & 1 & 1\end{bmatrix}$, $\c{b} = \begin{bmatrix} 0.5 \\[0.3em] 0.5 \\[0.3em] 0.5 \\[0.3em] 2\end{bmatrix}$, and  
$\c{A} = \begin{bmatrix}  1 & 1 & 1 & 1 & 0 & 0 & 0 & 0 & 0 & 0 & 0 & 0 \\[0.3em] 0 & 0 & 0 & 0 & 1 & 1 & 1 & 1 & 0 & 0 & 0 & 0 \\[0.3em] 0 & 0 & 0 & 0 & 0 & 0 & 0 & 0 & 1 & 1 & 1 & 1 \\[0.3em] 1 & 2 & 3 & 4 & 1 & 2 & 3 & 4 & 1 & 2 & 3 & 4\end{bmatrix}$.
Notice the integrality gap has been reflected in defining $\c{b}$.

\hspace{2cm}\\
\textbf{Initialization Step}

Let the starting basis consist of $\c{s}$ and $\lambda_0$ where $\X_0=\vec{0}$ is the starting integer point. 
Therefore, the first simplex tableau is as the following.

\begin{tabular}{l |*{5}{P{.4cm}}|c}  
\multicolumn{1}{r|}{} & \multicolumn{5}{c|}{\bi} & \multicolumn{1}{c}{\rhs} \\ 
\hline 
$z$ & $0$ & $0$ & $0$ & $0$ & $0$ & $0$ \\ 
\hline 
$s_1$ & $1$ & $0$ & $0$ & $0$ & $0$ & $.5$ \\ 
$s_2$ & $0$ & $1$ & $0$ & $0$ & $0$ & $.5$ \\ 
$s_3$ & $0$ & $0$ & $1$ & $0$ & $0$ & $.5$ \\ 
$s_4$ & $0$ & $0$ & $0$ & $1$ & $0$ & $2$ \\ 
$\lambda_0$ & $0$ & $0$ & $0$ & $0$ & $1$ & $1$ \\ 
\end{tabular}

\hspace{3cm}\\
\textbf{Iteration 1}

\subprob.
From the simplex tableau, we have $\c{w}=\begin{bmatrix}0 & 0 & 0 & 0 \end{bmatrix}$ and $\alpha=0$. 
As a result, $\c{wA}+\c{c}=\c{c}$.
The subproblem therefore is $\max_{j \in \mathcal{I}} \c{c}\X_j$. 
Subroutine \ax\ returns $\X$ such that $\X_{11}=\X_{22}=1$ and all other entries of $\X$ are zero.
The objective of the point with respect to current cost is $z-\hat{c}=10>0$.
Let us call this point $\X_1$.

\masprob.

$\c{A}\X_1=\begin{bmatrix} 1 \\[0.3em] 1 \\[0.3em] 0 \\[0.3em] 3 \end{bmatrix}$. Then $\c{y}_1=\c{B}^{-1}\begin{bmatrix} \c{A}\X_1 \\[0.3em] 1 \end{bmatrix}=\begin{bmatrix} 1 \\[0.3em] 1 \\[0.3em] 0 \\[0.3em] 3 \\[0.3em] 1\end{bmatrix}$.

Now, we insert the column into the foregoing tableau and pivot.
Variable $s_1$ leaves the basis and $\lambda_1$ enters the basis.

\hspace{1cm}\\
\begin{tabular}{l |*{5}{P{.4cm}}|c|c}  
\multicolumn{1}{r|}{} & \multicolumn{5}{c|}{\bi} & \multicolumn{1}{c|}{\rhs} & \multicolumn{1}{c}{$\lambda_1$}\\ 
\hline 
$z$ & $0$ & $0$ & $0$ & $0$ & $0$ & $0$ & $10$\\ 
\hline 
$s_1$ & $1$ & $0$ & $0$ & $0$ & $0$ & $.5$ & $\c{\underline{1}}$ \\ 
$s_2$ & $0$ & $1$ & $0$ & $0$ & $0$ & $.5$ & $1$ \\ 
$s_3$ & $0$ & $0$ & $1$ & $0$ & $0$ & $.5$ & $0$ \\ 
$s_4$ & $0$ & $0$ & $0$ & $1$ & $0$ & $2$ & $3$ \\ 
$\lambda_0$ & $0$ & $0$ & $0$ & $0$ & $1$ & $1$ & $1$ \\ 
\end{tabular}

\hspace{1cm}\\
After pivoting we obtain the following tableau.

\hspace{1cm}\\
\begin{tabular}{l | *{5}{P{.7cm}} |c}  
\multicolumn{1}{r|}{} & \multicolumn{5}{c|}{\bi} & \multicolumn{1}{c}{\rhs} \\ 
\hline 
$z$ & $-10$ & $0$ & $0$ & $0$ & $0$ & $-5$ \\ 
\hline 
$\lambda_1$ & $1$ & $0$ & $0$ & $0$ & $0$ & $.5$ \\ 
$s_2$ & $-1$ & $1$ & $0$ & $0$ & $0$ & $0$ \\ 
$s_3$ & $0$ & $0$ & $1$ & $0$ & $0$ & $.5$ \\ 
$s_4$ & $-3$ & $0$ & $0$ & $1$ & $0$ & $.5$ \\ 
$\lambda_0$ & $-1$ & $0$ & $0$ & $0$ & $1$ & $.5$ \\ 
\end{tabular}

\hspace{2cm}\\
The best-known feasible solution of the overall problem is given by
$\lambda_0 \X_0 + \lambda_1 \X_1=0.5 \X_0 + 0.5 \X_1$.
The current objective value is $5$.

\hspace{3cm}\\
\textbf{Iteration 2}

\subprob.
From the simplex tableau, we have $\c{w}=\begin{bmatrix}-10 & 0 & 0 & 0 \end{bmatrix}$ and $\alpha=0$. 
As a result, $\c{wA}+\c{c}=\begin{bmatrix}-4 & -4 & -4 & -4 & 1 & 4 & 4 & 6 & 0 & 1 & 1 & 1\end{bmatrix}$.
The subproblem therefore is $\max_{j \in \mathcal{I}} (\c{wA}+\c{c})\X_j$. 
Subroutine \ax\ returns $\X$ such that $\X_{24}=1$ and all other entries of $\X$ are zero.
The objective of the point with respect to current cost is $z-\hat{c}=6>0$.
Let us call this point $\X_2$.

\masprob.

$\c{A}\X_2=\begin{bmatrix} 0 \\[0.3em] 1 \\[0.3em] 0 \\[0.3em] 4 \end{bmatrix}$. Then $\c{y}_2=\c{B}^{-1}\begin{bmatrix} \c{A}\X_2 \\[0.3em] 1 \end{bmatrix}=\begin{bmatrix} 0 \\[0.3em] 1 \\[0.3em] 0 \\[0.3em] 4 \\[0.3em] 1\end{bmatrix}$.

Now, we insert the column into the foregoing tableau and pivot.
Variable $s_2$ leaves the basis and $\lambda_2$ enters the basis.

\hspace{1cm}\\
\begin{tabular}{l | *{5}{P{.7cm}} |c|c}  
\multicolumn{1}{r|}{} & \multicolumn{5}{c|}{\bi} & \multicolumn{1}{c|}{\rhs} & \multicolumn{1}{c}{$\lambda_2$}\\ 
\hline 
$z$ & $-10$ & $0$ & $0$ & $0$ & $0$ & $-5$ & $6$ \\ 
\hline 
$\lambda_1$ & $1$ & $0$ & $0$ & $0$ & $0$ & $.5$ & $0$ \\ 
$s_2$ & $-1$ & $1$ & $0$ & $0$ & $0$ & $0$ & $\underline{\c{1}}$ \\ 
$s_3$ & $0$ & $0$ & $1$ & $0$ & $0$ & $.5$ & $0$ \\ 
$s_4$ & $-3$ & $0$ & $0$ & $1$ & $0$ & $.5$ & $4$ \\ 
$\lambda_0$ & $-1$ & $0$ & $0$ & $0$ & $1$ & $.5$ & $1$ \\ 
\end{tabular}

\hspace{1cm}\\
After pivoting we obtain the following tableau.

\hspace{1cm}\\
\begin{tabular}{l | *{5}{P{.7cm}} |c}  
\multicolumn{1}{r|}{} & \multicolumn{5}{c|}{\bi} & \multicolumn{1}{c}{\rhs} \\ 
\hline 
$z$ & $-4$ & $-6$ & $0$ & $0$ & $0$ & $-5$ \\ 
\hline 
$\lambda_1$ & $1$ & $0$ & $0$ & $0$ & $0$ & $.5$ \\ 
$\lambda_2$ & $-1$ & $1$ & $0$ & $0$ & $0$ & $0$ \\ 
$s_3$ & $0$ & $0$ & $1$ & $0$ & $0$ & $.5$ \\ 
$s_4$ & $1$ & $-4$ & $0$ & $1$ & $0$ & $.5$ \\ 
$\lambda_0$ & $0$ & $-1$ & $0$ & $0$ & $1$ & $.5$ \\ 
\end{tabular}

\hspace{2cm}\\
The best-known feasible solution of the overall problem is given by
$\lambda_0 \X_0 + \lambda_1 \X_1=0.5 \X_0 + 0.5 \X_1$.
The current objective value is $5$.

\hspace{3cm}\\
\textbf{Iteration 3}

\subprob.
From the simplex tableau, we have $\c{w}=\begin{bmatrix}-4 & -6 & 0 & 0 \end{bmatrix}$ and $\alpha=0$. 
As a result, $\c{wA}+\c{c}=\begin{bmatrix}2 & 2 & 2 & 2 & -5 & -2 & -2 & 0 & 0 & 1 & 1 & 1\end{bmatrix}$.
The subproblem therefore is $\max_{j \in \mathcal{I}} (\c{wA}+\c{c})\X_j$. 
Subroutine \ax\ returns $\X$ such that $\X_{11}=1$, $\X_{32}=1$ and all other entries of $\X$ are zero.
The objective of the point with respect to current cost is $z-\hat{c}=3>0$.
Let us call this point $\X_3$.

\masprob.

$\c{A}\X_3=\begin{bmatrix} 1 \\[0.3em] 0 \\[0.3em] 1 \\[0.3em] 3 \end{bmatrix}$. Then $\c{y}_2=\c{B}^{-1}\begin{bmatrix} \c{A}\X_3 \\[0.3em] 1 \end{bmatrix}=\begin{bmatrix} 1 \\[0.3em] -1 \\[0.3em] 1 \\[0.3em] 4 \\[0.3em] 1\end{bmatrix}$.

Now, we insert the column into the foregoing tableau and pivot.
Variable $s_4$ leaves the basis and $\lambda_3$ enters the basis.

\hspace{1cm}\\
\begin{tabular}{l | *{5}{P{.7cm}} |c|c}  
\multicolumn{1}{r|}{} & \multicolumn{5}{c|}{\bi} & \multicolumn{1}{c|}{\rhs} & \multicolumn{1}{c}{$\lambda_3$}\\ 
\hline 
$z$ & $-4$ & $-6$ & $0$ & $0$ & $0$ & $-5$ & $3$\\ 
\hline 
$\lambda_1$ & $1$ & $0$ & $0$ & $0$ & $0$ & $.5$ & $1$\\ 
$\lambda_2$ & $-1$ & $1$ & $0$ & $0$ & $0$ & $0$ & $-1$\\ 
$s_3$ & $0$ & $0$ & $1$ & $0$ & $0$ & $.5$ & $1$ \\ 
$s_4$ & $1$ & $-4$ & $0$ & $1$ & $0$ & $.5$ & $\c{\underline{4}}$\\ 
$\lambda_0$ & $0$ & $-1$ & $0$ & $0$ & $1$ & $.5$ & $1$\\ 
\end{tabular}

\hspace{1cm}\\
After pivoting we obtain the following tableau.

\hspace{1cm}\\
\begin{tabular}{l | *{5}{P{.7cm}} |c}  
\multicolumn{1}{r|}{} & \multicolumn{5}{c|}{\bi} & \multicolumn{1}{c}{\rhs} \\ 
\hline 
$z$ & $-4.75$ & $-3$ & $0$ & $-.75$ & $0$ & $-5.375$ \\ 
\hline 
$\lambda_1$ & $.75$ & $1$ & $0$ & $-.25$ & $0$ & $.375$ \\ 
$\lambda_2$ & $-.75$ & $0$ & $0$ & $.25$ & $0$ & $.125$ \\ 
$s_3$ & $-.25$ & $1$ & $1$ & $-.25$ & $0$ & $.375$ \\ 
$\lambda_3$ & $.25$ & $-1$ & $0$ & $.25$ & $0$ & $.125$ \\ 
$\lambda_0$ & $-.25$ & $0$ & $0$ & $-.25$ & $1$ & $.375$ \\ 
\end{tabular}

\hspace{2cm}\\
The best-known feasible solution of the overall problem is given by
$\lambda_0 \X_0 + \lambda_1 \X_1 + \lambda_2 \X_2 + \lambda_3 \X_3 = 0.375 \X_0 + 0.375 \X_1 + 0.125 \X_2 + 0.125 \X_3$.
The current objective value is $5.375$.

\hspace{3cm}\\
\textbf{Iteration 4}

\subprob.
From the simplex tableau, we have $\c{w}=\begin{bmatrix}-4.75 & -3 & 0 & -.75 \end{bmatrix}$ and $\alpha=0$. 
As a result, \\$\c{wA}+\c{c}=\begin{bmatrix}.5 & -.25 & -1 & -1.75 & -2.75 & -.5 & -1.25 & 0 & -.75 & -.5 & -1.25 & -2\end{bmatrix}$.
The subproblem therefore is $\max_{j \in \mathcal{I}} (\c{wA}+\c{c})\X_j$. 
Subroutine \ax\ returns $\X$ such that $\X_{11}=1$ and all other entries of $\X$ are zero.
The objective of the point with respect to current cost is $z-\hat{c}=0.5>0$.
Let us call this point $\X_4$.

\masprob.

$\c{A}\X_4=\begin{bmatrix} 1 \\[0.3em] 0 \\[0.3em] 0 \\[0.3em] 1 \end{bmatrix}$. Then $\c{y}_2=\c{B}^{-1}\begin{bmatrix} \c{A}\X_4 \\[0.3em] 1 \end{bmatrix}=\begin{bmatrix} .5 \\[0.3em] -.5 \\[0.3em] -.5 \\[0.3em] .5 \\[0.3em] .5\end{bmatrix}$.

Now, we insert the column into the foregoing tableau and pivot.
Variable $\lambda_3$ leaves the basis and $\lambda_4$ enters the basis.

\hspace{1cm}\\
\begin{tabular}{l | *{5}{P{.7cm}} |c|c}  
\multicolumn{1}{r|}{} & \multicolumn{5}{c|}{\bi} & \multicolumn{1}{c|}{\rhs} & \multicolumn{1}{c}{$\lambda_4$}\\ 
\hline 
$z$ & $-4.75$ & $-3$ & $0$ & $-.75$ & $0$ & $-5.375$ & $.5$\\ 
\hline 
$\lambda_1$ & $.75$ & $1$ & $0$ & $-.25$ & $0$ & $.375$ & $.5$\\ 
$\lambda_2$ & $-.75$ & $0$ & $0$ & $.25$ & $0$ & $.125$ & $-.5$\\ 
$s_3$ & $-.25$ & $1$ & $1$ & $-.25$ & $0$ & $.375$ & $-.5$ \\ 
$\lambda_3$ & $.25$ & $-1$ & $0$ & $.25$ & $0$ & $.125$ & $\c{\underline{.5}}$\\ 
$\lambda_0$ & $-.25$ & $0$ & $0$ & $-.25$ & $1$ & $.375$ & $.5$\\ 
\end{tabular}

\hspace{1cm}\\
After pivoting we obtain the following tableau.

\hspace{1cm}\\
\begin{tabular}{l | *{5}{P{.7cm}} |c}  
\multicolumn{1}{r|}{} & \multicolumn{5}{c|}{\bi} & \multicolumn{1}{c}{\rhs} \\ 
\hline 
$z$ & $-5$ & $-2$ & $0$ & $-1$ & $0$ & $-5.5$ \\ 
\hline 
$\lambda_1$ & $.5$ & $-2$ & $0$ & $-.5$ & $0$ & $.25$ \\ 
$\lambda_2$ & $-.5$ & $-1$ & $0$ & $.5$ & $0$ & $.25$ \\ 
$s_3$ & $0$ & $0$ & $1$ & $0$ & $0$ & $.5$ \\ 
$\lambda_4$ & $.5$ & $-2$ & $0$ & $.5$ & $0$ & $.25$ \\ 
$\lambda_0$ & $-.5$ & $1$ & $0$ & $-.5$ & $1$ & $.25$ \\ 
\end{tabular}

\hspace{2cm}\\
The best-known feasible solution of the overall problem is given by
$\lambda_0 \X_0 + \lambda_1 \X_1 + \lambda_2 \X_2 + \lambda_4 \X_4 = 0.25 \X_0 + 0.25 \X_1 + 0.25 \X_2 + 0.25 \X_4$.
The current objective value is $5.5$.

\hspace{3cm}\\
\textbf{Iteration 5}

\subprob.
From the simplex tableau, we have $\c{w}=\begin{bmatrix}-5 & -2 & 0 & -1 \end{bmatrix}$ and $\alpha=0$. 
As a result, \\$\c{wA}+\c{c}=\begin{bmatrix}-6 & -7 & -8 & -9 & -3 & -4 & -5 & -6 & -1 & -2 & -3 & -4\end{bmatrix}$.
The subproblem therefore is $\max_{j \in \mathcal{I}} (\c{wA}+\c{c})\X_j$. 
Subroutine \ax\ returns $\X=\vec{0}$.
The objective of the point with respect to current cost is $z-\hat{c}=0$. 
Therefore, the algorithm terminates.
Our final solution is as follows.
\begin{displaymath}
\c{x}^*= \begin{bmatrix} x_{11} \\[0.3em] x_{22}\\[0.3em] x_{24} \end{bmatrix}
= 0.25\begin{bmatrix} 1 \\[0.3em] 1 \\[0.3em] 0\end{bmatrix}
+ 0.25\begin{bmatrix} 0 \\[0.3em] 0 \\[0.3em] 1\end{bmatrix}
+ 0.25\begin{bmatrix} 1 \\[0.3em] 0 \\[0.3em] 0\end{bmatrix}
+ 0.25\begin{bmatrix} 0 \\[0.3em] 0 \\[0.3em] 0\end{bmatrix}
=\begin{bmatrix} 0.5 \\[0.3em] 0.25 \\[0.3em] 0.25\end{bmatrix}.\end{displaymath}

A simple examination shows that $\c{x}^*$ is in fact one half (scaled down by the integrality gap) of the optimal solution to $\text{MU-P}$ for our instance.


\end{document}